\documentclass[letterpaper, 10 pt, conference]{ieeeconf}  

\IEEEoverridecommandlockouts                              

\usepackage{graphics} 
\usepackage{epsfig} 
\usepackage{amsmath} 
\usepackage{amssymb}  
\usepackage{mathtools}
\usepackage{bbm}
\usepackage{dsfont}
\usepackage{enumerate}
\newtheorem{theorem}{Theorem}

\newtheorem{definition}{Definition}

\title{\LARGE \bf
Upper Bounds for Continuous-Time End-to-End Risks \\ in Stochastic Robot Navigation 
}

\author{ Apurva Patil$^1$ \and Takashi Tanaka$^2$ 
\thanks{*This work is supported by Lockheed Martin Corporation and FOA-AFRL-AFOSR-2019-0003.} \thanks{$^{1}$Walker Department of Mechanical Engineering, University of Texas at Austin. {\tt\small apurvapatil@utexas.edu}. $^{2}$Department of Aerospace Engineering and Engineering Mechanics, University of Texas at Austin. {\tt\small ttanaka@utexas.edu}.}
}

\begin{document}

\maketitle
\thispagestyle{empty}
\pagestyle{empty}

\begin{abstract}
We present an analytical method to estimate the continuous-time collision probability of motion plans for autonomous agents with linear controlled It\^{o} dynamics. Motion plans generated by planning algorithms cannot be perfectly executed by autonomous agents in reality due to the inherent uncertainties in the real world. Estimating end-to-end risk is crucial to characterize the safety of trajectories and plan risk optimal trajectories. In this paper, we derive upper bounds for the continuous-time risk in stochastic robot navigation using the properties of Brownian motion as well as Boole and Hunter's inequalities from probability theory. Using a ground robot navigation example, we numerically demonstrate that our method is considerably faster than the na\"ive Monte Carlo sampling method and the proposed bounds perform better than the discrete-time risk bounds.
\end{abstract}

\section{Introduction}\label{Sec: Intro}
Motion plans for mobile robots in obstacle-filled environments can be generated by autonomous trajectory planning algorithms \cite{lavalle2006planning}. In reality, due to the presence of uncertainties, the robots cannot follow the planned trajectories perfectly, and collisions with obstacles occur with a non-zero probability, in general. To address this issue, risk-aware motion planning has received considerable attention \cite{pepy2006safe}, \cite{blackmore2011chance}. Optimal planning under set-bounded uncertainty provides some solutions against worst-case disturbances \cite{majumdar2013robust}, \cite{lopez2019dynamic}.  However, in many cases, modeling uncertainties with unbounded (e.g. Gaussian) distributions has a number of advantages over a set-bounded approach \cite{blackmore2011chance}. In the case of unbounded uncertainties, it is generally difficult to guarantee safety against all realizations of noise. This motivates for an efficient risk estimation technique that can both characterize the safety of trajectories and be embedded in the planning algorithms to allow explicit trade-offs between control optimality and safety. In this paper, we develop an analytical method of continuous-time risk estimation for autonomous agents with linear controlled It\^{o} dynamics of the form (\ref{ito process}). We assume that a planned trajectory with a finite length in a known configuration space $\mathcal{X}\subseteq \mathbb{R}^n$ is given and a robot tracks this trajectory in finite time $T$. If $\pmb{x}^{sys}(t)\in\mathcal{X}$ represents the robot's position at time $t$, and $\mathcal{X}_{obs}\subset\mathcal{X}$ is the obstacle region, then the continuous-time end-to-end risk $\mathcal{R}$ in the navigation of the given trajectory can be written as
\begin{equation} \label{continuous risk}
 \mathcal{R}=P\left(\bigcup\limits_{t\in[0, T]}\pmb{x}^{sys}(t)\in\mathcal{X}_{obs}\right). 
 \end{equation}
Unfortunately, exact evaluation of (\ref{continuous risk}) is a challenging task because all the states $\pmb{x}^{sys}(t)$ across the time horizon $[0, T]$ are correlated with each other. In this paper, we derive two upper bounds for $\mathcal{R}$ by leveraging properties of Brownian motion (also called a Wiener process) as well as Boole and Hunter's inequalities from probability theory.\par
Monte Carlo and other sampling based methods \cite{blackmore2010probabilistic}, \cite{janson2018monte} provide accurate estimates of (\ref{continuous risk}) by computing the ratio of the number of simulated executions that collide with obstacles. However, these methods are often computationally expensive due to the need for a large number of simulation runs to obtain reliable estimates and are cumbersome to embed in planning algorithms.\par
The discrete-time risk estimation methods compute risks at the discretized time steps $t_i$, $i=0, 1, \hdots , N$, and approximate the probability in (\ref{continuous risk}) by
\begin{equation} \label{discrete risk}
 \mathcal{R}\approx P\left(\bigcup\limits_{i=0}^{N}\pmb{x}^{sys}(t_i)\in\mathcal{X}_{obs}\right).
 \end{equation}
Since the states $\{\pmb{x}^{sys}(t_i)\}_{i=0, 1, \hdots, N}$ are correlated with each other, evaluating the joint probability (\ref{discrete risk}) exactly is computationally expensive \cite{patil2021upper}. Several approaches have been proposed in the literature to upper bound this joint probability \cite{blackmore2011chance}, \cite{patil2021upper}, \cite{ono2015chance}. The commonly used approach is to use \textit{Boole's inequality} (a.k.a. union bound) which states that for any number of events $\mathcal{E}_j$, we have
\begin{equation}\label{Boole's inequality}
 P\left(\bigcup\limits_{j=1}^{N}\mathcal{E}_j\right)\leq \sum\limits_{j=1}^{N}P\left(\mathcal{E}_j\right).   
\end{equation}
Using this inequality, the probability in (\ref{discrete risk}) can be decomposed the over timesteps as \cite{blackmore2011chance}:
\begin{equation}\label{discrete-time Boole}
\mathcal{R}\leq\sum\limits_{i=0}^{N}P\left(\pmb{x}^{sys}(t_i)\in\mathcal{X}_{obs}\right).    
\end{equation}
While the discrete-time risk estimation approaches can be applied for continuous-time systems, their performance is highly sensitive to the chosen time discretization. They may underestimate the risk when the sampling rate is low or may produce severely conservative estimates when the sampling rate is high \cite{patil2021upper}.\par
Various continuous-time risk estimation approaches also have been proposed in the literature such as the approaches based on stochastic control barrier functions \cite{santoyo2021barrier}, \cite{yaghoubi2020risk}, cumulative lyapunov exponent \cite{oguri2019convex}, and first-exit times \cite{shah2011probability}, \cite{frey2020collision}, \cite{ariu2017chance}, \cite{chern2021safe}. The analyses presented by Shah et al. \cite{shah2011probability} and Chern et al. \cite{chern2021safe} give the exact continuous-time collision probability as the solution to a partial differential equation (PDE). Shah et al. \cite{shah2011probability} presents an analytic solution of this PDE for a simple case; namely that of a constrained spherical environment with no internal obstacles. However, such a closed-form solution is generally not tractable for complicated configuration spaces. Frey et al. \cite{frey2020collision} uses an interval-based integration scheme to approximate the collision probability by leveraging classical results in the study of first-exit times. Ariu et al. \cite{ariu2017chance} proposes an upper-bound for the continuous-time risk using the \textit{reflection principle} of Brownian motion and Boole's inequality (\ref{Boole's inequality}). In this paper, we extend the results presented in \cite{ariu2017chance} and derive tighter continuous-time risk bounds.\par
The contributions of this work are summarized as follows: We first use the \textit{Markov property} of Brownian motion, and tighten the risk bound derived in \cite{ariu2017chance}. We then further reduce the conservatism of this bound by leveraging \textit{Hunter's inequality} of the probability of union of events. Both our bounds possess the time-additive structure required in several optimal control techniques (e.g. dynamic programming) \cite{ono2015chance}, \cite{van2012motion}, making these bounds useful for risk-aware motion planning. Finally, using a ground robot navigation example, we demonstrate that our method requires considerably less computation time than the na\"ive Monte Carlo sampling method. We also show that compared to the discrete-time risk bound (\ref{discrete-time Boole}), our bounds are tighter, and at the same time ensure conservatism (i.e. safety).
\section{Preliminaries and Problem Statement}\label{Sec: PROBLEM FORMULATION}
\subsection{Planned Trajectory}
Let $\mathcal{X}_{free}=\mathcal{X}\backslash\mathcal{X}_{obs}$ be the obstacle-free region, and $\mathcal{X}_{goal}\subset \mathcal{X}$ be the target region. We assume that, for an initial position $x_0^{plan}\in\mathcal{X}_{free}$ of the robot, a trajectory planner gives us finite sequences of positions $\{x_j^{plan}\in\mathcal{X}_{free}\}_{j=0, 1, \hdots, N}$ and control inputs $\{v_j^{plan}\in\mathbb{R}^n\}_{j=0, 1, \hdots, N-1}$ such that $x_N^{plan}\in\mathcal{X}_{goal}$. Let $\mathcal{T}=\left(0=t_{0}< t_1<\hdots<t_{N}=T\right)$ be the partition of the time horizon $[0, T]$, with $\Delta t_j=t_{j+1}-t_j$ satisfying \begin{equation}\label{Delta tj}
    v_j^{plan}\Delta t_j=x_{j+1}^{plan}-x_j^{plan}, \qquad j=0, 1, \hdots, N-1.
\end{equation}
The \textit {planned trajectory}, $x^{plan}(t)$, $t\in[0, T]$ is generated by the linear interpolations between $x_{j}^{plan}$ and $x_{j+1}^{plan}$, $(j=0, 1, \hdots, N-1)$. 
\par
\subsection{Robot Dynamics}\label{Sec: Robot Dynamics}
Assume that a robot following the planned path generates a trajectory defined by a random process $\pmb{x}^{sys}(t)$, $t\in[0, T]$ with associated probability space $\left(\Omega, \mathcal{F}, P\right)$. We assume that the process $\pmb{x}^{sys}(t)$ satisfies the following controlled It\^{o} process:
\begin{equation}\label{ito process}
    d\pmb{x}^{sys}(t)=\pmb{v}^{sys}(t)dt+R^{\frac{1}{2}}d\pmb{w}(t),\qquad t\in[0, T]
\end{equation}
with $\pmb{x}^{sys}(0)=x^{plan}_0$. Here, $\pmb{v}^{sys}(t)$ is the velocity input command, $\pmb{w}(t)$ is the $n$-dimensional standard Brownian motion, and $R$ is a given positive definite matrix used to model the process noise intensity. We assume that the robot tracks the planned trajectory in open-loop using a piecewise constant control input:
\begin{equation}\label{v_true(t)}
    \pmb{v}^{sys}(t)=v_j^{plan} \qquad \forall\; t\in[t_j, t_{j+1}).
\end{equation}
 \par
The time discretization of (\ref{ito process}) under $\mathcal{T}$, based on the Euler-Maruyama method \cite{kloeden1992stochastic} yields:
\begin{equation}\label{discretization of ito}
    \pmb{x}^{sys}(t_{j+1})=\pmb{x}^{sys}(t_j)+\pmb{v}^{sys}(t_j)\Delta t_j+\pmb{n}(t_j)
\end{equation}
where $\pmb{n}(t_j)\sim\mathcal{N}(0, \Delta t_jR)$. Introducing $\pmb{x}^{sys}_j\coloneqq \pmb{x}^{sys}(t_j)$, $\pmb{u}^{sys}_j\coloneqq \pmb{v}^{sys}(t_j)\Delta t_j$, $\pmb{n}_j\coloneqq \pmb{n}(t_j)$, and $\Sigma_{\pmb{n}_j}\coloneqq \Delta t_jR$, (\ref{discretization of ito}) can be rewritten as
\begin{equation}\label{xj(true) dynamics}
    \pmb{x}^{sys}_{j+1}=\pmb{x}^{sys}_j+\pmb{u}^{sys}_j+\pmb{n}_j, \qquad \pmb{n}_j\sim\mathcal{N}(0, \Sigma_{\pmb{n}_j}),
\end{equation}
for $j=0, 1, \hdots, N-1$. Further, using (\ref{v_true(t)}) and (\ref{Delta tj}), $\pmb{u}_j^{sys}$ can be rewritten as \begin{equation}\label{uj_true}
 \pmb{u}_j^{sys}=v_j^{plan}\Delta t_j=x_{j+1}^{plan}-x_j^{plan}.   
\end{equation}
\par
Let 
\vspace{-4mm}
\begin{equation}\label{deviation trajectory}
    \pmb{x}(t)\coloneqq \pmb{x}^{sys}(t)-x^{plan}(t),\qquad t\in[0, T]
\end{equation}
be the deviation of the robot from the planned trajectory during trajectory tracking. Defining $\pmb{x}_j\coloneqq \pmb{x}(t_j)$, from (\ref{xj(true) dynamics}), (\ref{uj_true}) and (\ref{deviation trajectory}), the dynamics of $\pmb{x}_j$ can be written as
\begin{equation}\label{xj dynamics}
    \pmb{x}_{j+1}=\pmb{x}_j+\pmb{n}_j, \qquad \pmb{n}_j\sim\mathcal{N}(0, \Sigma_{\pmb{n}_j})
\end{equation}
for $j=0, 1, \hdots, N-1$ with $\pmb{x}_0=0$. The state $\pmb{x}_j$ is distributed as $\pmb{x}_j\sim\mathcal{N}(0, \Sigma_{\pmb{x}_j})$ where $\Sigma_{\pmb{x}_j}$ propagates according to $\Sigma_{\pmb{x}_{j+1}}=\Sigma_{\pmb{x}_j}+\Sigma_{\pmb{n}_j}, j=0, 1, \hdots, N-1$, with the initial covariance $\Sigma_{\pmb{x}_0}=0$.\par
\subsection{Problem Statement}
As explained in Section \ref{Sec: Intro}, the continuous-time end-to-end risk $\mathcal{R}$ over the time horizon $[0, T]$ is formulated as (\ref{continuous risk}). Under $\mathcal{T}$, we reformulate $\mathcal{R}$ as follows:
\begin{equation}\label{union of seg}
\mathcal{R}=\!P\left(\bigcup\limits_{j=1}^{N}\bigcup\limits_{t\in \mathcal{T}_j}\pmb{x}^{sys}(t)\!\in\mathcal{X}_{obs}\right)
\end{equation}
where $\mathcal{T}_j= [t_{j-1}, t_j]$, $j=1, 2, \hdots, N$. In the rest of the paper, we deal with formulation (\ref{union of seg}) in order to derive upper bounds for $\mathcal{R}$.
\subsection{Properties of Brownian Motion} \label{Sec: Known Results of a one-dimensional Brownian motion}

\begin{definition}[Markov property] \label{TM: Markov property}
Let $\pmb{w}(t)$, $t\geq0$ be an $n$-dimensional Brownian motion started in $z\in \mathbb{R}^n$. Let $s\geq0$, then the process $\widetilde{\pmb{w}}(t)\coloneqq\pmb{w}(t+s)-\pmb{w}(s)$, $t\geq0$ is again a Brownian motion started in the origin and it is independent of the process $\pmb{w}(t)$, $0\leq t\leq s$.

\end{definition}
\begin{theorem}[Reflection principle]\label{TH: reflection principle}
If $\pmb{w}(t)$, $t\geq 0$ is a one-dimensional Brownian motion started in the origin and $d>0$ is a threshold value, then
\begin{equation}\label{reflection}
    P\left(\underset{s\in[0,t]}{\text{sup}}\pmb{w}(s)\geq d\right)=2P\left(\pmb{w}(t)\geq d\right).
\end{equation}
\end{theorem}
Refer to \cite{durrett2019probability} and \cite{morters2010brownian} for the proof.
\par
\section{Continuous-Time Risk Analysis}
In this Section, we first reformulate $\mathcal{R}$ in terms of one-dimensional Brownian motions and then use the properties from Section \ref{Sec: Known Results of a one-dimensional Brownian motion} to compute bounds for $\mathcal{R}$. For the analysis in Sections \ref{Sec: Risk Analysis Problem in terms of a one-dimensional Brownian motion} to \ref{Sec: Second-Order Risk Bound}, we assume that $\mathcal{X}_{obs}$ is convex. In Section \ref{Sec: multi polyhedrons}, we explain how the analysis can be generalized when $\mathcal{X}_{obs}$ is non-convex. 
\subsection{$\mathcal{R}$ in terms of One-Dimensional Brownian Motions}\label{Sec: Risk Analysis Problem in terms of a one-dimensional Brownian motion}
Let $\mathcal{S}_j$ be the path segment connecting $x_{j-1}^{plan}$ and $x_{j}^{plan}$ or equivalently, $x^{plan}(t_{j-1})$ and $x^{plan}(t_j)$, $j=1, 2, \hdots, N$. Now, we conservatively approximate $\mathcal{X}_{obs}$ with a half space, similar to \cite{morgan2014model}, \cite{zhu2019chance}. Since $\mathcal{S}_j$ and $\mathcal{X}_{obs}$ are convex, bounded and disjoint subsets of $\mathbb{R}^n$, from the hyperplane separation theorem, we can guarantee the existence of a hyperplane that strictly separates $\mathcal{S}_j$ and $\mathcal{X}_{obs}$. Let $\mathcal{H}_j: a_j^Tx-b_j=0$, $a_j\in\mathbb{R}^n$, $b_j\in\mathbb{R}$, $\|a_j\|=1$ be a hyperplane such that $\mathcal{X}_{obs}\subseteq \mathcal{H}_j^+$ and $\mathcal{S}_j\subset\mathcal{H}_j^-$ where the half spaces $\mathcal{H}_j^+$ and $\mathcal{H}_j^-$ are defined as
\begin{equation}\label{Hj+-}
    \mathcal{H}_j^+\coloneqq\{x\in\mathbb{R}^n: a_j^Tx-b_j\geq0\}, \quad \mathcal{H}_j^-\coloneqq\mathbb{R}^n\backslash\mathcal{H}_j^+.
\end{equation}  
 Since $\mathcal{H}_j^+$ is a conservative approximation of $\mathcal{X}_{obs}$, we can upper bound $\mathcal{R}$ in (\ref{union of seg}) as
\begin{equation}\label{R in H_j^+}
\mathcal{R}\leq\!P\left(\bigcup\limits_{j=1}^{N}\bigcup\limits_{t\in \mathcal{T}_j}\pmb{x}^{sys}(t)\!\in\mathcal{H}_j^+\right).
\end{equation}
To find a least conservative upper bound, each hyperplane $\mathcal{H}_j$ can be constructed using the solution ($y_1^*$, $y_2^*$) to the following optimization problem:
\begin{equation}\label{optimization problem}
\begin{aligned}
\min_{y_1, y_2\in\mathbb{R}^n} \quad & \|y_1-y_2\|\\
\textrm{s.t.} \quad & y_1\in \mathcal{X}_{obs}, \; y_2\in \mathcal{S}_j.
\end{aligned}    
\end{equation}
The least conservative hyperplane $\mathcal{H}_j$ will be perpendicular to the line segment connecting $y_1^*$ and $y_2^*$, and passing through $y_1^*$. If $d_j\coloneqq\|y_1^*-y_2^*\|$, then $d_j$ represents the minimum distance of $\mathcal{S}_j$ from $\mathcal{X}_{obs}$. Fig. \ref{Fig. conversion to 1D} shows an example of an optimal hyperplane $\mathcal{H}_j$ for a given $\mathcal{X}_{obs}$ and $\mathcal{S}_j$. 
\begin{figure}[t]
    \centering
    \includegraphics[scale=0.5]{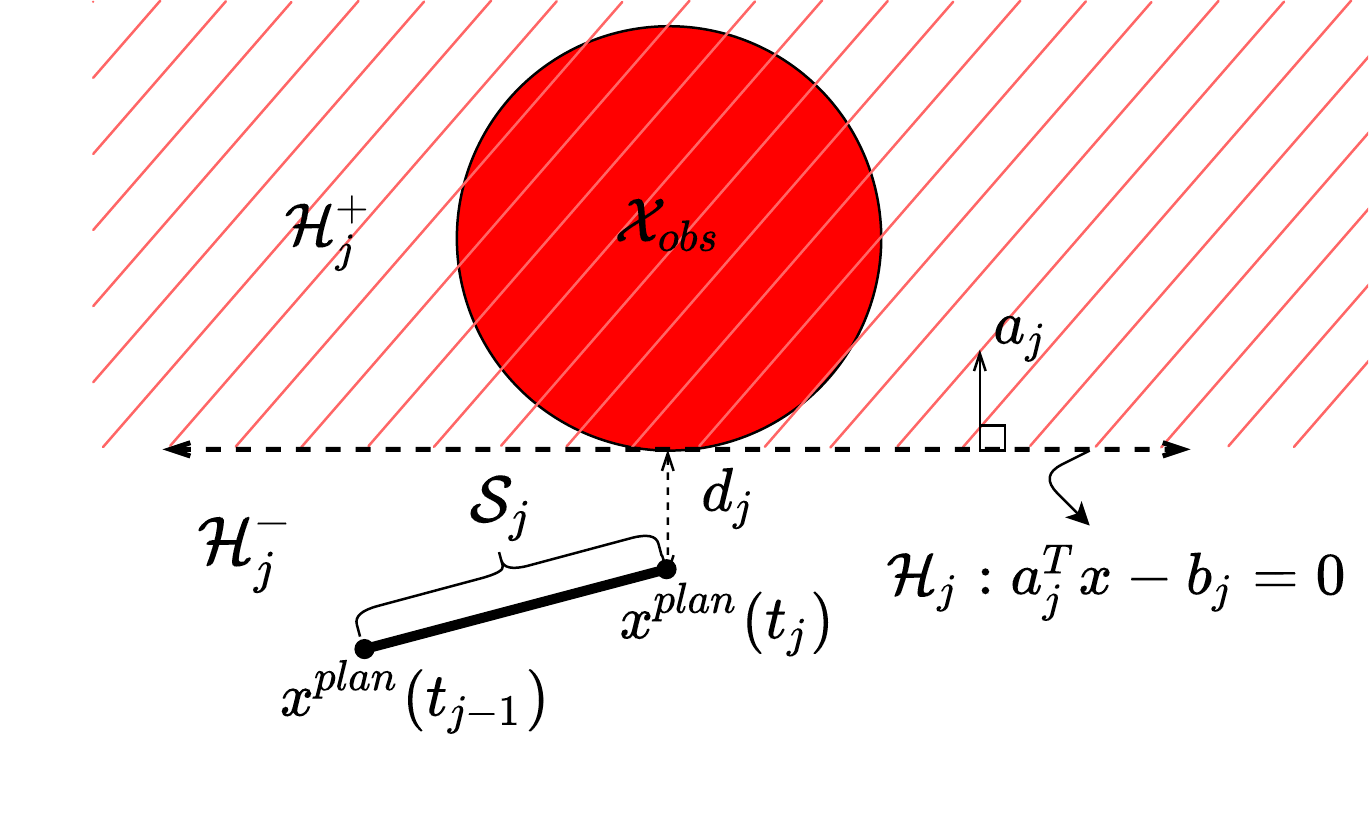}
    \caption{The least conservative hyperplane $\mathcal{H}_j$ approximating $\mathcal{X}_{obs}$ (shown in a red-faced circle) with a half space $\mathcal{H}_j^+$ (shown in red hatching). $d_j$ is the minimum distance of $\mathcal{S}_j$ from $\mathcal{X}_{obs}$.}
    \label{Fig. conversion to 1D}
\end{figure}\par
Now, it can be shown that  
\begin{equation}\label{Hj in d_j}
\left(\bigcup\limits_{t\in\mathcal{T}_j}\pmb{x}^{sys}(t)\in \mathcal{H}_j^+\!\! \right)\subseteq\left(\bigcup\limits_{t\in\mathcal{T}_j}a_j^T\pmb{x}(t)\geq d_j\!\right)
\end{equation}
where $\pmb{x}(t)$ is the deviation of the robot from the planned trajectory as defined in (\ref{deviation trajectory}). Proof of (\ref{Hj in d_j}) is presented in Appendix A. Using (\ref{R in H_j^+}) and (\ref{Hj in d_j}), $\mathcal{R}$ can be upper-bounded as 
\begin{equation}\label{R in dj}
\mathcal{R}\leq\!P\left(\bigcup\limits_{j=1}^{N}\bigcup\limits_{t\in \mathcal{T}_j}a_j^T\pmb{x}(t)\geq d_j\right).
\end{equation}
For the proposed robot dynamics (Section \ref{Sec: Robot Dynamics}), it is trivial to show that $a_j^T\pmb{x}(t)$ is a one-dimensional Brownian motion for $t\in[0, T]$ that starts in the origin. Let us denote $\pmb{w}_j(t)\coloneqq a_j^T\pmb{x}(t)$, $j=1, 2, \hdots, N$. Now, (\ref{R in dj}) can be written as 
\begin{equation}\label{R in Bj}
    \mathcal{R}\leq P\left(\bigcup\limits_{j=1}^{N}\;\underset{t\in\mathcal{T}_j}{\text{max}}\;\pmb{w}_j(t)\geq d_j\right).
\end{equation}
Defining $\mathcal{E}_j\coloneqq\left(\underset{t\in\mathcal{T}_j}{\text{max}}\;\pmb{w}_j(t)\geq d_j\right)$, (\ref{R in Bj}) can be rewritten as
\begin{equation}\label{R in Ej}
    \mathcal{R}\leq P\left(\bigcup\limits_{j=1}^{N}\mathcal{E}_j\right).
\end{equation}
Since $\{\mathcal{E}_j\}_{j=1, 2, \hdots, N}$ are non-independent events, computing (\ref{R in Ej}) exactly is a challenging task. In the following sections, we derive bounds for $P\left(\bigcup\limits_{j=1}^{N}\mathcal{E}_j\right)$. 
\subsection{First-Order Risk Bound}\label{Sec: First-Order Risk Bound}
Define $p_j\coloneqq P(\mathcal{E}_j)=P\left( \underset{t\in[t_{j-1}, t_j]}{\text{max}}\pmb{w}_j(t)\geq d_j\right)$. Applying Boole's inequality (\ref{Boole's inequality}), the probability in (\ref{R in Ej}) can be decomposed as 
\begin{equation}\label{1st order bound}
\mathcal{R}\leq P\left(\bigcup\limits_{j=1}^{N}\mathcal{E}_j\right)\leq\sum\limits_{j=1}^{N}p_j.
\end{equation}
This gives us a first-order risk bound for $\mathcal{R}$. $p_j$ is the continuous-time risk associated with the time segment $\mathcal{T}_j=[t_{j-1}, t_j]$. Note that the bound in (\ref{1st order bound}) possesses the time-additive structure which is helpful to use this bound in the risk-aware motion planning algorithms. \par
In order to take advantage of the reflection principle to compute $p_j$, Ariu et al. \cite{ariu2017chance} proposes to compute an upper bound to $p_j$ as
\begin{equation}\label{pj Kaito}
    p_j\leq P\left(\underset{t\in[0, t_j]}{\text{max}}\,\pmb{w}_j(t)\geq d_j\right).
\end{equation}
Using the reflection principle (\ref{reflection}), the right side of (\ref{pj Kaito}) can be evaluated as
\begin{equation}\label{Kaito reflection}
 \!\!P\!\left(\!\underset{t\in[0, t_j]}{\text{max}}\!\pmb{w}_j(t)\!\geq\!d_j\!\!\right)\!\!=\! 2P\!\left(\pmb{w}_j(t_j)\!\geq\! d_j\!\right)\!=\!2P\!\left(\!a_j^T\pmb{x}_j\!\geq\!d_j\!\right)\!.
\end{equation}
From (\ref{1st order bound}), (\ref{pj Kaito}), and (\ref{Kaito reflection}) we get
\begin{equation}\label{Kaito's bound}
\begin{aligned}
    \mathcal{R}\leq2\sum\limits_{j=1}^{N}P\left(a_j^T\pmb{x}_j\geq d_j\right).
\end{aligned}
\end{equation}
The bound in (\ref{Kaito's bound}) requires computing probabilities only at the discrete-time steps, simplifying the estimation of the continuous-time risk. However, the over-approximation in (\ref{pj Kaito}) introduces unnecessary conservatism that can be avoided using the Markov property of Brownian motion. Next, we present a way by which $p_j$ can be computed exactly without any over-approximation. \par
For notational convenience, let us denote the random variables $\pmb{w}_j(t_{j-1})$ and $\pmb{w}_j(t_{j})$ by $\pmb{z}^{s}_j$ and $\pmb{z}^{e}_j$ respectively:
\begin{equation}\label{z_js, z_je}
    \pmb{z}^{s}_j\coloneqq \pmb{w}_j(t_{j-1})=a_j^T\pmb{x}_{j-1},\quad \pmb{z}^{e}_j\coloneqq \pmb{w}_j(t_{j})=a_j^T\pmb{x}_{j}
\end{equation}
for $j=1, 2, \hdots, N$. If $\mu_{\pmb{\xi}}(\xi)$ denotes the probability density function (p.d.f.) of any random variable $\pmb{\xi}$, then  
\begin{equation}\label{mu of zjs, zje}
\begin{aligned}
\mu_{\pmb{z}^{s}_j}({z}^{s}_j)=\mathcal{N}(0, \sigma_{\pmb{z}^{s}_j}^2),& \qquad \sigma_{\pmb{z}^{s}_j}^2=a_j^T\Sigma_{\pmb{x}_{j-1}}a_j,\\
\mu_{\pmb{z}^{e}_j}({z}^{e}_j)=\mathcal{N}(0, \sigma_{\pmb{z}^{e}_j}^2),& \qquad \sigma_{\pmb{z}^{e}_j}^2=a_j^T\Sigma_{\pmb{x}_j}a_j.
\end{aligned}
\end{equation}
Let us define $\pmb{z}_j\coloneqq\begin{bmatrix}\pmb{z}^{s}_j & \pmb{z}^{e}_j\end{bmatrix}^T\in\mathbb{R}^2$. It is straightforward to show that the joint distribution of $\pmb{z}_j$ is
\begin{equation}\label{pdf of z_j}
\mu_{\pmb{z}_j}({z}_j)=\mathcal{N}\left(0, \Sigma_{\pmb{z}_j}\right),\qquad\Sigma_{\pmb{z}_j}=\begin{bmatrix}
    \sigma_{\pmb{z}^{s}_j}^2 & \sigma_{\pmb{z}^{s}_j}^2\\
   \sigma_{\pmb{z}^{s}_j}^2 & \sigma_{\pmb{z}^{e}_j}^2
    \end{bmatrix}.
\end{equation}
\par
Now, we compute $p_j$ using the following theorem: 
\begin{theorem}\label{TH: pj}
If $\mu_{\pmb{z}^{s}_j}({z}^{s}_j)$ and $\mu_{\pmb{z}_j}({z}_j)$ are the distributions of the normal random variables $\pmb{z}^{s}_j$ and $\pmb{z}_j$, represented as (\ref{mu of zjs, zje}) and (\ref{pdf of z_j}) respectively, then $p_j$ is given by:
\begin{equation}\label{theorem pj}
\begin{aligned}
   p_j\!=\!\!\!\int_{{z}_j^s=d_j}^{\infty}\!\!\!\!\!\mu_{\pmb{z}^{s}_j}({z}^{s}_j)d{z}_j^s\!+2\!\!\int_{{z}_j^s=-\infty}^{d_j}\!\int_{{z}_j^e=d_j}^{\infty}\!\!\!\!\mu_{\pmb{z}_j}({z}_j)d{z}_j^{e}d{z}_j^{s}.
\end{aligned}
\end{equation}
\end{theorem}
\begin{proof}
Let us define:
\begin{equation*}
   \begin{aligned}
    p_j^1\coloneqq P\left( \underset{t\in[t_{j-1}, t_j]}{\text{max}}\pmb{w}_j(t)\geq d_j,\; \pmb{w}_j(t_{j-1})\geq d_j\right),\\
    p_j^2\coloneqq P\left( \underset{t\in[t_{j-1}, t_j]}{\text{max}}\pmb{w}_j(t)\geq d_j,\; \pmb{w}_j(t_{j-1})< d_j\right).
   \end{aligned}
\end{equation*}
Using the law of total probability, we can write $p_j$ as
\begin{equation}\label{split pj}
p_j=p_j^1+p_j^2.
\end{equation}
Since $\left(\pmb{w}_j(t_{j-1})\geq d_j\right)\subseteq \left(\underset{t\in[t_{j-1}, t_j]}{\text{max}}\pmb{w}_j(t)\geq d_j\right)$, $p_j^1$ can be computed as
\begin{equation}\label{final pj1}
    \begin{aligned}
        p_j^1=P\left(\pmb{w}_j(t_{j-1})\geq d_j\right)=\int_{{z}_j^s=d_j}^{\infty}\mu_{\pmb{z}^{s}_j}({z}^{s}_j)d{z}_j^s.
    \end{aligned}
\end{equation}
Now, we write $p_j^2$ as 
\begin{equation*}
    \begin{aligned}
        p_j^2=&P\left( \underset{t\in[t_{j-1}, t_j]}{\text{max}}\pmb{w}_j(t)\geq d_j,\; \pmb{z}_j^s< d_j\right)\\
        =&P\left(\underset{t\in[0, t_j-t_{j-1}]}{\text{max}}\pmb{w}_j(t+t_{j-1})\geq d_j,\;\pmb{z}_j^s<d_j\right).
    \end{aligned}
\end{equation*}
From Markov property of Brownian motion (Definition \ref{TM: Markov property}), 
\begin{equation}\label{B_tilde}
    \widetilde{\pmb{w}}_j(t)=\pmb{w}_j(t+t_{j-1})-\pmb{w}_j(t_{j-1}), \quad t\in[0, (T-t_{j-1})]
\end{equation}
is a one-dimensional Brownian motion that starts in the origin. Rewriting $p_j^2$ in terms of $\widetilde{\pmb{w}}_j(t)$, we get
\begin{equation}\label{E_j in B_tilde_j}
    \begin{aligned}
    p_j^2&=P\left(\underset{t\in[0, t_j-t_{j-1}]}{\text{max}}\widetilde{\pmb{w}}_j(t)\geq d_j-\pmb{z}_j^{s},\; \pmb{z}_j^{s}<d_j\right)\\
&=\!\!\!\int_{-\infty}^{d_j}\!\!\!P\!\left(\underset{t\in[0, t_j-t_{j-1}]}{\text{max}}\!\widetilde{\pmb{w}}_j(t)\!\geq\! d_j-{z}_j^{s}\right)\mu_{\pmb{z}^{s}_j}({z}^{s}_j)d{z}_j^{s}.
    \end{aligned}
\end{equation}
Since $d_j-{z}_j^{s}>0$, $\forall\,{z}_j^s\in(-\infty, d_j)$, we can apply the reflection principle (\ref{reflection}) and rewrite (\ref{E_j in B_tilde_j}) as
\begin{equation}\label{E_j in B_tilde_j2}
\begin{aligned}
p_j^2\!=\!\!\int_{-\infty}^{d_j}\!\!\!2P\!\left(\widetilde{\pmb{w}}_j(t_j\!-\!t_{j-1})\geq d_j\!-\!{z}_j^{s}\right)\mu_{\pmb{z}^{s}_j}({z}^{s}_j)d{z}_j^{s}.
\end{aligned}
\end{equation}
Let us denote the random variable $\widetilde{\pmb{w}}_j(t_j-t_{j-1})$ by $\pmb{y}_j$. Using (\ref{B_tilde}) and (\ref{z_js, z_je}), 
\begin{equation*}
    \pmb{y}_j\coloneqq \widetilde{\pmb{w}}_j(t_j-t_{j-1})=\pmb{w}_j(t_j)-\pmb{w}_j(t_{j-1})=\pmb{z}_j^{e}-\pmb{z}_j^{s},
\end{equation*}
and the p.d.f. of $\pmb{y}_j$ is $\mu_{\pmb{y}_j}({y}_j)=\mathcal{N}(0, \sigma_{\pmb{y}_j}^2)$ where $\sigma_{\pmb{y}_j}^2=\sigma_{\pmb{z}^{e}_j}^2-\sigma_{\pmb{z}^{s}_j}^2$.
Now, (\ref{E_j in B_tilde_j2}) can be rewritten as 
\begin{equation}\label{pj|Fcj}
\begin{aligned}
    p_j^2&=2\int_{-\infty}^{d_j}\left(\int_{d_j-{z}_j^{s}}^{\infty}\mu_{\pmb{y}_j}({y}_j) \,d{y}_j\right)\mu_{\pmb{z}^{s}_j}({z}^{s}_j)d{z}_j^{s}\\
    &\!\!\!\!\!\!\!\!\!\!\!\!\!\!=\!2\!\!\!\int_{\!\!-\infty}^{d_j}\!\!\int_{\!d_j-{z}_j^{s}}^{\infty}\!\frac{1}{2\pi\sigma_{\pmb{z}_j^{s}}\sigma_{\pmb{y}_j}}exp\!\!\left\{\!\!-\frac{1}{2}\!\!\left(\!\!\frac{{z}_j^s}{\sigma_{\pmb{z}_j^s}}\!\!\right)^{\!\!\!2}\!\!-\!\frac{1}{2}\!\!\left(\!\!\frac{{y}_j}{\sigma_{\pmb{y}_j}}\!\!\right)^{\!\!\!2}\!\!\right\}\!\!d{y}_jd{z}_j^{s}.
\end{aligned}     
\end{equation}
The outside integral in right side of (\ref{pj|Fcj}) is w.r.t. ${z}_j^s$ and the inside is one is w.r.t. ${y}_j$. Substituting ${y}_j$ with ${z}_j^{e}-{z}_j^{s}$, (\ref{pj|Fcj}) can be rewritten as 
\begin{equation}\label{bivariate integral}
\begin{aligned}
    p_j^2=&2\!\!\int_{{z}_j^s=-\infty}^{d_j}\!\int_{{z}_j^e=d_j}^{\infty}\frac{1}{2\pi\sigma_{\pmb{z}_j^{s}}\sigma_{\pmb{z}_j^{e}}\sqrt{1-\rho^2}}\cdot\\
    &\!\!\!\!\!\!\!\!\!\!\!\!\!\!\!\!\!\!\!\!\!\!\!\!exp\!\left\{\!\!-\frac{1}{2\!\left(1\!-\!\rho^2\right)}\!\!\left[\left(\!\!\frac{{z}_j^s}{\sigma_{\pmb{z}_j^s}}\!\!\right)^{\!\!\!2}\!\!-\frac{2\rho\,{z}_j^s{z}_j^e}{\sigma_{\pmb{z}_j^s}\sigma_{\pmb{z}_j^e}}+\!\!\left(\!\!\frac{{z}_j^e}{\sigma_{\pmb{z}_j^e}}\!\!\right)^{\!\!\!2}\right]\right\}d{z}_j^{e}\,d{z}_j^{s}
\end{aligned}
\end{equation}
where $\rho=\sigma_{\pmb{z}_j^s}/\sigma_{\pmb{z}_j^e}$. The expression inside the double integral of (\ref{bivariate integral}) is a bivariate normal distribution of $\pmb{z}_j$. Hence,
\begin{equation}\label{final pj2}
p_j^2=2\!\!\int_{{z}_j^s=-\infty}^{d_j}\int_{{z}_j^e=d_j}^{\infty}\mu_{\pmb{z}_j}({z}_j)d{z}_j^{e}\,d{z}_j^{s}.    
\end{equation}
Combining (\ref{split pj}), (\ref{final pj1}) and (\ref{final pj2}) we recover (\ref{theorem pj}).
\end{proof}\par
 MATLAB's \texttt{mvncdf} function can be utilized to compute the integrations (\ref{final pj1}) and (\ref{final pj2}) numerically.\par
\subsection{Second-Order Risk Bound}\label{Sec: Second-Order Risk Bound}
The proposed first-order risk bound (\ref{1st order bound}) can be tightened using a variant of Hunter's inequality that additionally considers the joint probability of consecutive events \cite{prekopa2003probabilistic}:
\begin{equation*}\label{2nd order risk bound}
\mathcal{R}\leq P\left(\bigcup\limits_{j=1}^{N}\mathcal{E}_j\right)\leq \sum\limits_{j=1}^{N} p_j-\sum\limits_{j=1}^{N-1} p_{j,\,j+1}
\end{equation*}
where $ p_{j,\,j+1}\coloneqq P(\mathcal{E}_j\cap \mathcal{E}_{j+1})$ is the joint risk associated with the time segments $\mathcal{T}_j$ and $\mathcal{T}_{j+1}$. Computing $p_{j,\,j+1}$ exactly is challenging. In this work, we propose to compute a lower bound $p_{j,\,j+1}^{LB}$ of $p_{j,\,j+1}$ using the following theorem:
\vspace{1mm}
\begin{theorem}\label{lower bound to pjk}
If $t_{j-1}= \hat{t}_j^0<\hat{t}_j^1<\hdots<\hat{t}_j^{r_j}=t_{j}$ is a discretization of the time segment $\mathcal{T}_j$, and $\pmb{z}_j^{i}$, $\mathcal{D}_j$ are defined as
\begin{equation*}
    \pmb{z}_j^{i}\coloneqq \pmb{w}_j(\hat{t}_j^i)= a_j^T\pmb{x}(\hat{t}_j^i),
\end{equation*}
\begin{equation*}\label{Dj in Bj}
\mathcal{D}_j\coloneqq\left(\pmb{z}_j^0<d_j\right)\cap\left(\pmb{z}_j^1<d_j\right)\cap\hdots\cap\left(\pmb{z}_j^{r_j}<d_j\right),
\end{equation*}
then $p_{j,\,j+1}$ is lower bounded by $p_{j,\,j+1}^{LB}$ given as
\begin{equation*}\label{pjkD}
p_{j,\,j+1}^{LB}=1-P(\mathcal{D}_j)-P(\mathcal{D}_{j+1})+P(\mathcal{D}_j\cap \mathcal{D}_{j+1}).
\end{equation*}
\end{theorem}
\vspace{2mm}
\begin{proof}
Introduce $\mathcal{C}_j$ as 
\begin{equation*}
\begin{aligned}
 \mathcal{C}_j=&\left(\pmb{w}_j(\hat{t}_j^0)\geq d_j\right)\cup\left(\pmb{w}_j(\hat{t}_j^1)\geq d_j\right)\cup\hdots\cup\left(\pmb{w}_j(\hat{t}_j^{r_j})\geq d_j\right)\\
=&\left(\pmb{z}_j^0\geq d_j\right)\cup\left(\pmb{z}_j^1\geq d_j\right)\cup\hdots\cup\left(\pmb{z}_j^{r_j}\geq d_j\right)\\
=&\mathcal{D}_j^c.
\end{aligned}    
\end{equation*}
Now, since $\mathcal{C}_j\subset \mathcal{E}_j$
\begin{equation*}
\begin{aligned}
    p_{j,\,j+1}&\geq P(\mathcal{C}_j\cap \mathcal{C}_{j+1})\\
    &=1-P(\mathcal{D}_j\cup \mathcal{D}_{j+1})\\
    &=1-P(\mathcal{D}_j)-P(\mathcal{D}_{j+1})+P(\mathcal{D}_j\cap \mathcal{D}_{j+1})\\
    &=p_{j,\,j+1}^{LB}.
\end{aligned}
\end{equation*}
\end{proof}
$P(\mathcal{D}_j)$ can be computed by finding the joint distribution of $\begin{bmatrix} \pmb{z}_j^{0}& \pmb{z}_j^{1}& \hdots &\pmb{z}_j^{r_j} \end{bmatrix}^T$ and $P(\mathcal{D}_j\cap \mathcal{D}_{j+1})$ by finding the joint distribution of $\begin{bmatrix} \pmb{z}_j^{0}& \pmb{z}_j^{1}& \hdots &\pmb{z}_j^{r_j}& \pmb{z}_{j+1}^{0}& \pmb{z}_{j+1}^{1}& \hdots &\pmb{z}_{j+1}^{r_{j+1}} \end{bmatrix}^T$. The computations of $P(\mathcal{D}_j)$ and $P(\mathcal{D}_j\cap \mathcal{D}_{j+1})$ are summarized in Appendix B. Now, we get our second-order risk bound as follows: 
\begin{equation} \label{2nd order risk bound in pL}
\mathcal{R}\leq \sum\limits_{j=1}^{N} p_j-\sum\limits_{j=1}^{N-1} p_{j,\,j+1}^{LB}.
\end{equation}
Similar to the first-order risk bound (\ref{1st order bound}), this bound also possesses the time-additive structure. Note that the higher sampling rates we choose to discretize the time segments $\mathcal{T}_j$ (a set of higher $r_j$'s), the tighter the bound in (\ref{2nd order risk bound in pL}) becomes. \par


\subsection{Risk Analysis when $\mathcal{X}_{obs}$ is Non-Convex}\label{Sec: multi polyhedrons}
As mentioned earlier, the analysis in Sections \ref{Sec: Risk Analysis Problem in terms of a one-dimensional Brownian motion} to \ref{Sec: Second-Order Risk Bound} assumes that $\mathcal{X}_{obs}$ is convex, which is sufficient to guarantee the existence of a set of separating hyperplanes $\{\mathcal{H}_j\}_{j=1, 2, \hdots, N}$. When $\mathcal{X}_{obs}$ is non-convex, we partition it into $M$ subregions $\mathcal{X}_{obs_m}$, $m=1, 2, \hdots, M$ such that $\mathcal{X}_{obs}=\left(\bigcup\limits_{m=1}^{M}\mathcal{X}_{obs_m}\right)$ and a set of separating hyperplanes $\{\mathcal{H}_j\}_{j=1, 2, \hdots, N}$ exists for each $\mathcal{X}_{obs_m}$. We then bound $\mathcal{R}$ as 
\begin{equation*}
\mathcal{R}\leq\sum\limits_{m=1}^{M}\mathcal{R}_m,\qquad \mathcal{R}_m=P\left(\bigcup\limits_{t\in[0, T]}\pmb{x}^{sys}(t)\in\mathcal{X}_{obs_m}\right).   
\end{equation*}
The first and second-order upper bounds for $\mathcal{R}_m$ can be computed using the analysis in Sections \ref{Sec: Risk Analysis Problem in terms of a one-dimensional Brownian motion} to \ref{Sec: Second-Order Risk Bound}. In order to obtain tight upper-bounds for $\mathcal{R}$, the partitioning of $\mathcal{X}_{obs}$ can be optimized which is left for the future work. 
\section{Simulation Results}\label{Sec: example}
In this Section, we demonstrate  the validity and performance of our continuous-time risk bounds via a ground robot navigation simulation. The configuration space is $\mathcal{X}=[0, 1]\times[0, 1]$. We assume that the robot dynamics are governed by the It\^o process (\ref{ito process}), with $R=10^{-3}\times I$ ($I$ is a $2\times2$ identity matrix), and it is commanded to travel at a unit velocity i.e., $\|\pmb{v}^{sys}(t)\|=1$, $t\in[0, T]$. As explained in Section \ref{Sec: PROBLEM FORMULATION}, we discretize the dynamics (\ref{ito process}) under the time partition $\mathcal{T}$. Due to the unit velocity assumption, $\Delta t_j=\|x_{j+1}^{plan}-x_j^{plan}\|$. Hence, our discrete-time robot dynamics are
\begin{equation}\label{discrete robot dynamics eg}
\pmb{x}^{sys}_{j+1}\!\!=\!\pmb{x}^{sys}_j\!\!+\pmb{u}^{sys}_j\!\!+\pmb{n}_j, \quad \pmb{n}_j\!\sim\!\mathcal{N}(0, \|x_{j+1}^{plan}-x_j^{plan}\|R),
\end{equation}
where $\pmb{u}_j^{sys}$ is defined as per (\ref{uj_true}). The model (\ref{discrete robot dynamics eg}) is natural for ground robots whose location uncertainty grows linearly with the distance traveled.


\par First, we plan trajectories using RRT* with the instantaneous safety criterion \cite{pedram2021rationally} (i.e., at every time step, the confidence ellipse with a fixed safety level is collision-free). For a given configuration space, four planned trajectories with $95\%$, $75\%$, $50\%$, and $25\%$ instantaneous safety levels are shown in Fig. \ref{Fig. planned trajectories}. In each case, the confidence ellipses grow in size with the distance since the robot tracks these trajectories in open-loop.
\begin{figure}[t]
    \centering
      \begin{tabular}{c c}
      \includegraphics[scale=0.25]{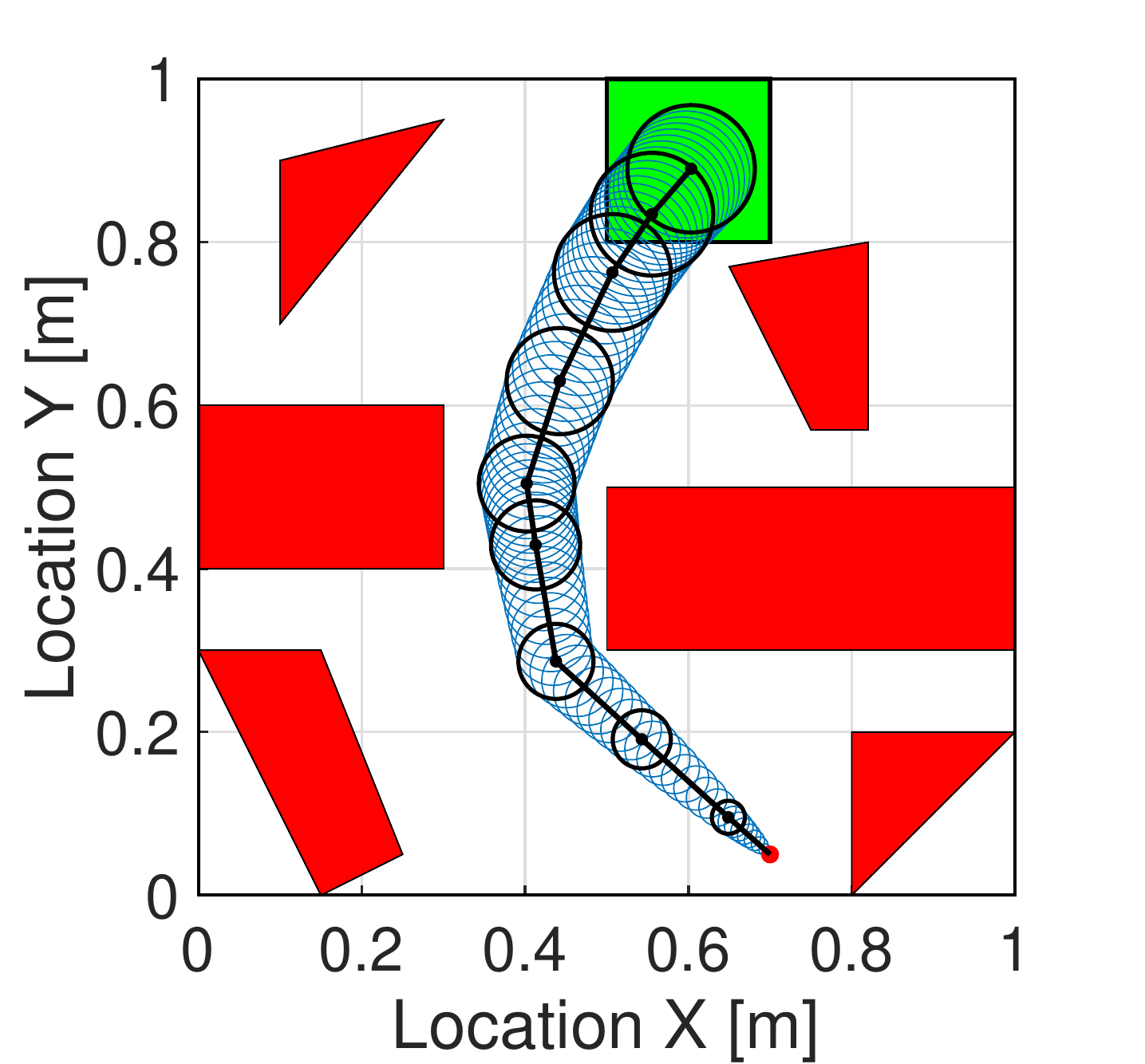} &\!\!\!\!\!\!\!\includegraphics[scale=0.25]{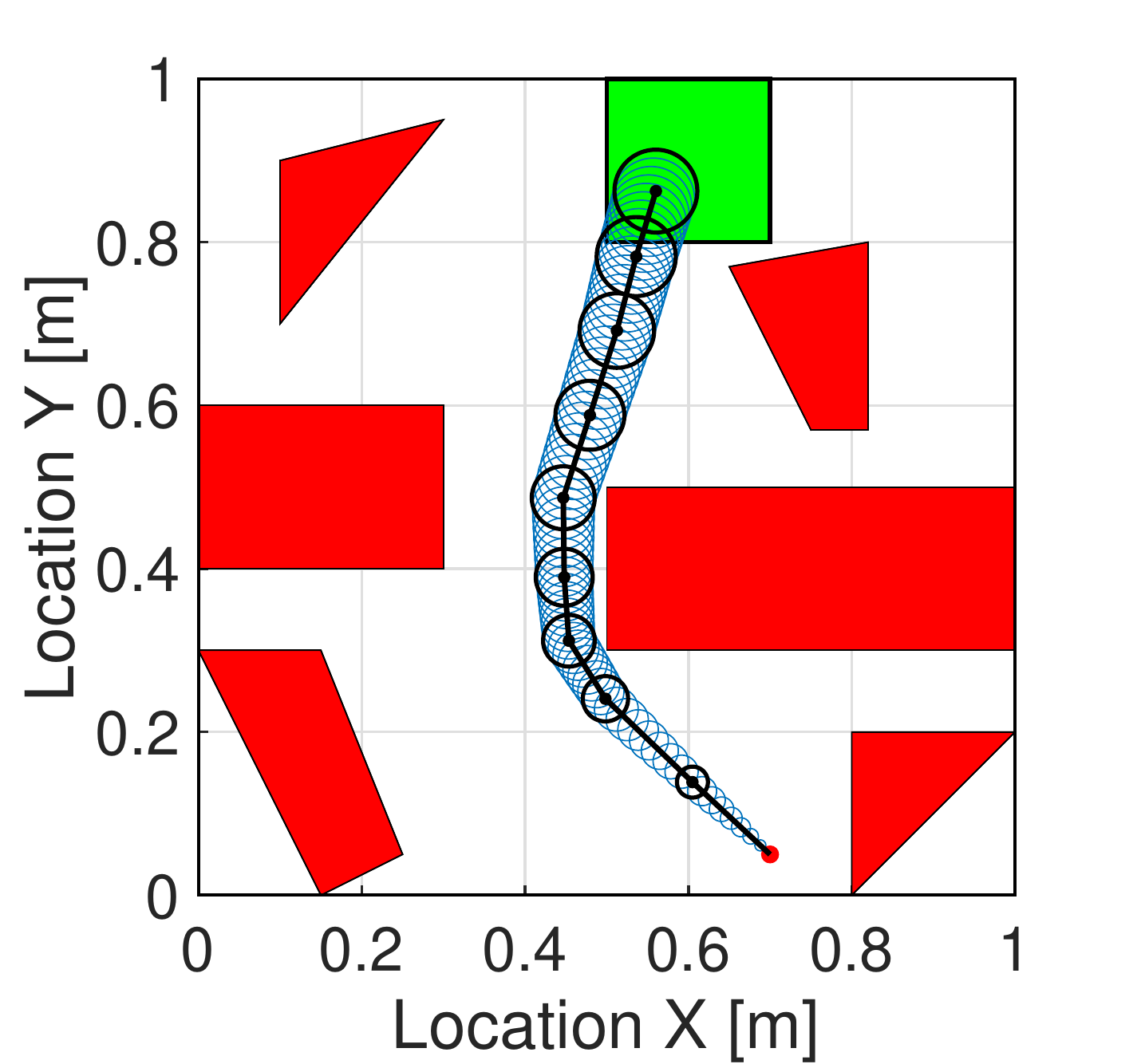} \\
      (a)&(b)\\
        \includegraphics[scale=0.25]{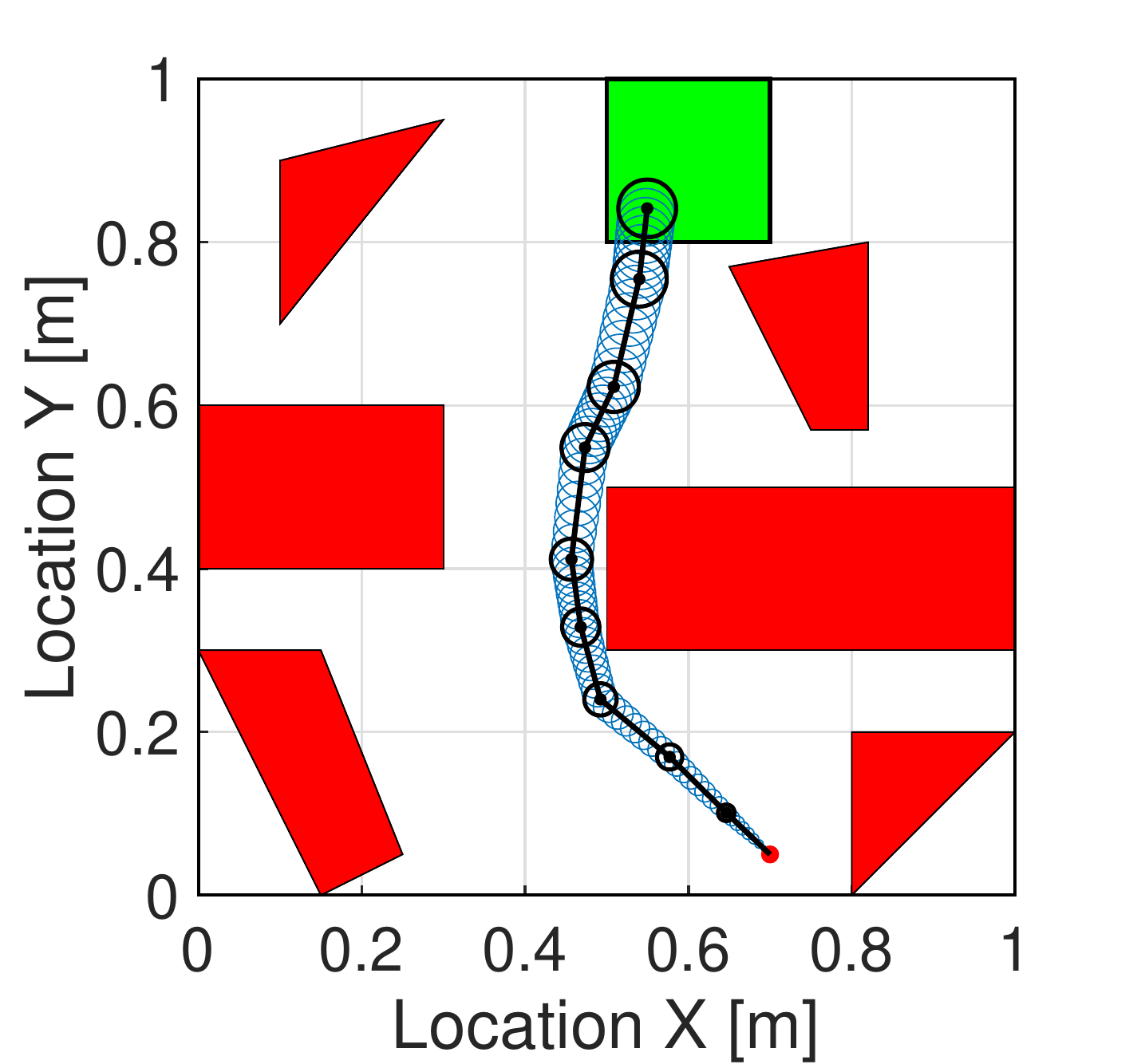} &\!\!\!\!\!\!\! \includegraphics[scale=0.25]{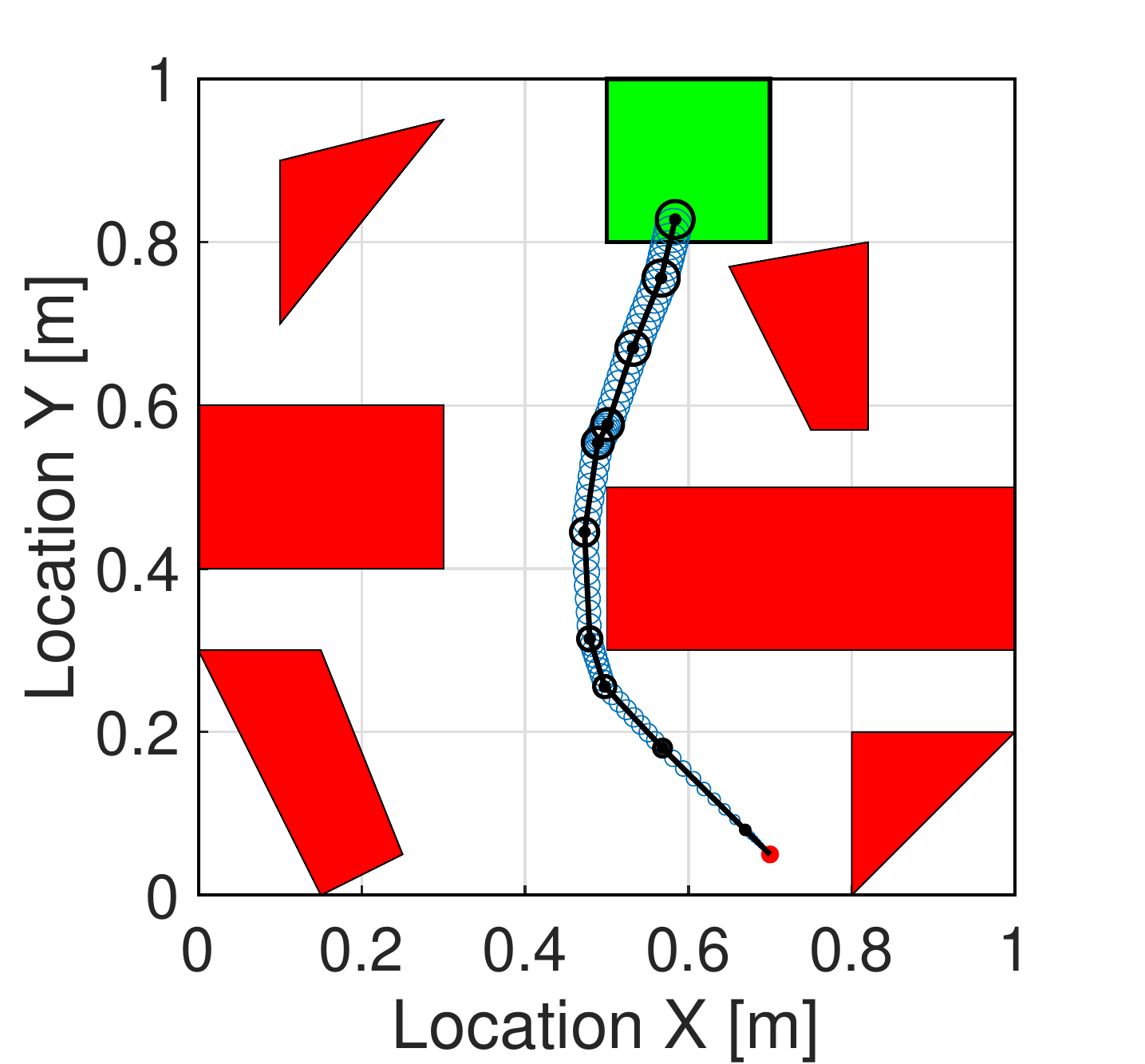}\\
      (c)&(d)
       \end{tabular}
        \caption{Trajectories planned with the instantaneous safety criterion \cite{pedram2021rationally} are shown in black. The black dots on the trajectories represent the planned positions $\{x^{plan}_j\}_{j=1, 2, \hdots, N}$. The red dot represents the initial position $x^{plan}_0$ of the robot. The red-faced polygons represent $\mathcal{X}_{obs}$ and the green faced rectangle represents $\mathcal{X}_{goal}$.  The trajectories are shown with (a) $95\%$, (b) $75\%$, (c) $50\%$ and (d) $25\%$ confidence ellipses. The ellipses at the time steps $\{t_j\}_{j=1, 2, \hdots, N}$ are shown in black and the ellipses at the intermediate time steps are shown in blue.} 
        \label{Fig. planned trajectories}
\end{figure}
Fig. \ref{Fig. all bounds} plots the continuous and discrete-time risk bounds for these plans having different instantaneous safety (risk) levels. For validation, we compute failure probabilities using $10^5$ Monte Carlo simulations at a high rate of time discretization ($r_d=100$) and assume them as the ground truths (shown in black). The dotted graphs are the discrete-time risk bounds ($B_d$) computed using (\ref{discrete-time Boole}) at different rates of time discretization ($r_d$). As is evident from the graph, the discrete-time risk bounds at a lower rate of time discretization underestimate the Monte Carlo estimates, and as the time-discretization rate increases, they become overly conservative. On the other hand, our continuous-time risk bounds ($B_c$) (shown with solid red and blue graphs) are tighter, and at the same time ensure conservatism.
\begin{figure}[t]
    \centering
    \includegraphics[scale=0.2]{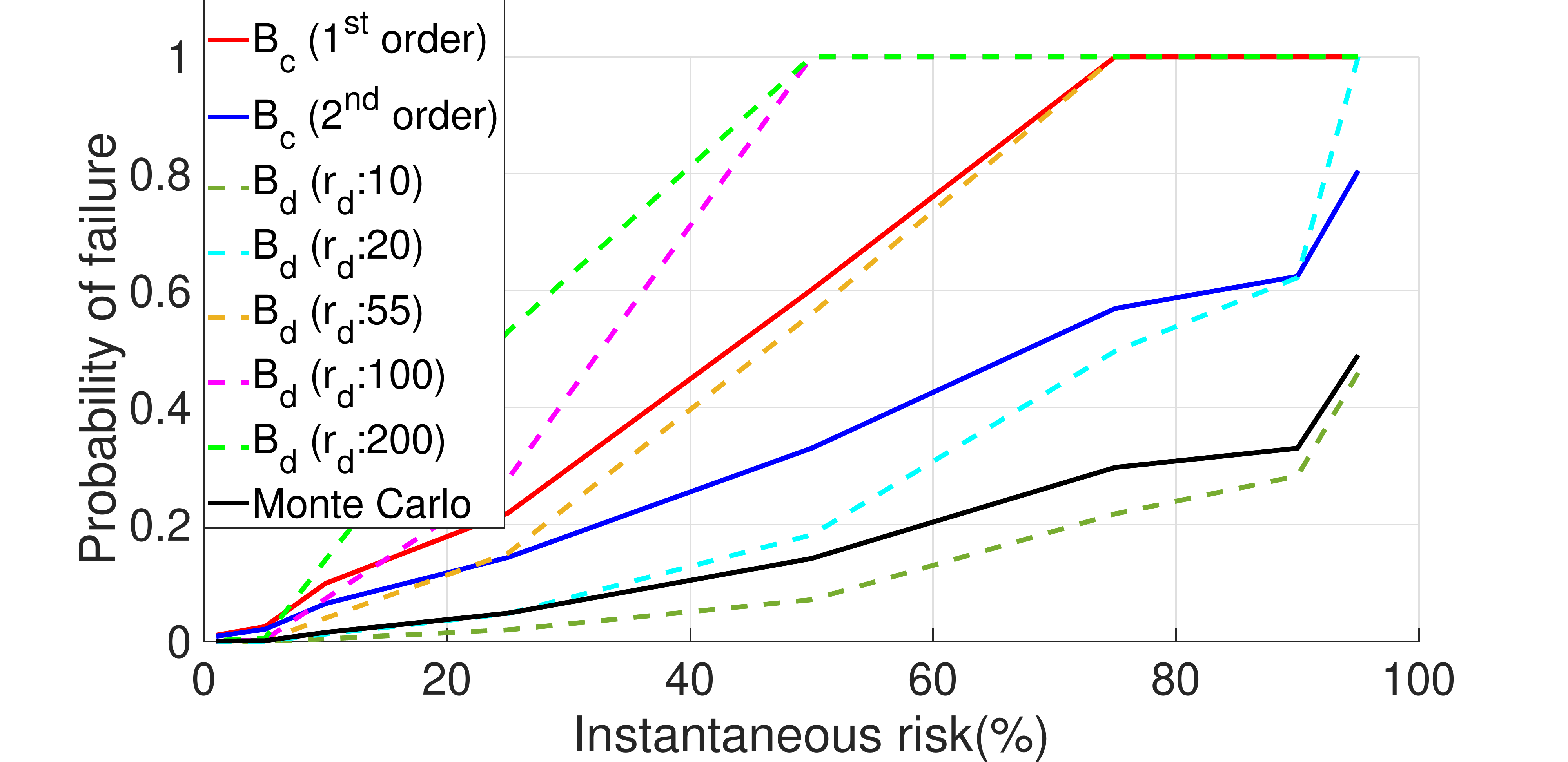}
    \caption{End-to-end probabilities of failure computed for the trajectories with different instantaneous risk levels. The solid red and blue graphs represent the first-order and second-order continuous-time risk bounds ($B_c$) respectively. The dotted graphs are discrete-time risk bounds ($B_d$) computed using (\ref{discrete-time Boole}) at different rates of time discretization ($r_d$). The Monte Carlo estimates of the same trajectories are shown in black.}
    \label{Fig. all bounds}
\end{figure}\par
Next, we demonstrate a larger statistical evaluation over $100$ trajectories planned using RRT* in randomly-generated environments (random initial, goal and obstacle positions). These trajectories are generated with $5\%$ instantaneous safety criterion \cite{pedram2021rationally}. The average risk estimate of $10^5$ Monte Carlo simulations (run at a high rate of time discretization $r_d=100$) is $0.27$. The statistics of the discrete-time and continuous-time risk estimates are shown in Table \ref{Table: comparison}. The discrete-time risk estimates are computed using (\ref{discrete-time Boole}) at increasing rates of time discretization ($r_d$). The continuous-time risk estimates are computed using the method proposed by Ariu et al. \cite{ariu2017chance} and our approach. The \textbf{Bias} and \textbf{RMSE} columns lists respectively the mean (signed) difference and the root mean squared difference between the corresponding estimate and the Monte Carlo estimate. The \textbf{\% Conservative} column reports the percentage of cases where the corresponding estimate was greater than (or within 0.1\% of) the Monte Carlo estimate and the \textbf{Avg. Time} lists the average computation times for our MATLAB implementations.     
\begin{table}[h]
\caption{Comparison of different risk estimates over $100$ trajectories. Computation is performed in MATLAB on a consumer laptop.}
\label{Table: comparison}
\begin{center}
\begin{tabular}{ |c|c|c|c|c| } 
 \hline
 \textbf{Risk Estimates} & \textbf{Avg. Time} & \textbf{Bias}  & \textbf{RMSE} & \textbf{\%Conservative}  \\ 
 \hline\hline
  Monte Carlo &101.50 s & 0 & 0 & - \\  
  \hline
  \textit{Discrete-time} &&&&\\
 $r_d: 5$ &0.14 s& -0.14 & 0.18 & 28\% \\
 $r_d: 10$ & 0.26 s& -0.002 & 0.16 & 59\% \\
 $r_d: 20$ &0.52 s& 0.31 &  0.57 & 82\% \\
 $r_d: 55$ &1.53 s& 1.50 & 2.33 & 100\% \\
 $r_d: 100$ &2.87 s& 2.98 &  4.53 & 100\% \\
 \hline
 \textit{Continuous-time}&&&&\\
 Ariu et al. \cite{ariu2017chance} &1.39 s& 0.97  & 1.33 & 100\% \\
 Our $1^{st}$ order&1.47 s& 0.66 &  0.90 & 100\% \\
 Our $2^{nd}$ order&2.23 s& 0.28 &  0.36 & 100\% \\
 \hline
\end{tabular}
\end{center}
\end{table}
From the data presented, following conclusions can be drawn: First, our risk bounds require significantly less computation time than the Monte Carlo method. Second, unlike the discrete-time risk bounds at the lower sampling rates, our bounds remain conservative (i.e., safe) in all the trials. Lastly, our bounds produce tighter estimates than the discrete-time risk bounds at the higher sampling rates and the continuous-time risk bound of \cite{ariu2017chance}.
\section{Conclusion}
In this paper, we conducted an analysis to estimate the continuous-time collision probability of motion plans for autonomous agents with linear controlled It\^{o} dynamics. We derived two upper bound for the continuous-time risk using the properties of Brownian motion (Markov property and reflection principle), and probability inequalities (Boole and Hunter's inequality). Our method boils down to computing probabilities at the discrete-time steps, simplifying the analysis, yet providing risk guarantees in continuous-time. We show that our bounds outperform the discrete-time risk bound (\ref{discrete-time Boole}) and are cheaper in computation than the na\"ive Monte Carlo sampling method.\par
Our analysis motivates a number of future investigations. This paper assumes that the robot follows a linear controlled It\^{o} process. Future work will focus on risk analysis for systems with generalized stochastic dynamics. Another direction we would like to explore is risk analysis by fusing sampling-based methods and methods from continuous stochastic processes as suggested in \cite{frey2020collision}. This hybrid approach may provide the best of both worlds: high accuracy as well as computational simplicity and compatibility with continuous optimization. 




\section*{APPENDIX A}
\section*{Proof of (\ref{Hj in d_j})}
\vspace{-5mm}
\begin{equation}\label{17_LHS}
    \begin{aligned}
    \left(\bigcup\limits_{t\in\mathcal{T}_j}\pmb{x}^{sys}(t)\in \mathcal{H}_j^+\!\! \right)&=\left(\bigcup\limits_{t\in\mathcal{T}_j}a_j^T\pmb{x}^{sys}(t)\geq b_j\right)\\
    &\!\!\!\!\!\!\!\!\!\!\!\!\!\!\!\!\!\!\!\!\!\!\!\!\!\!=\left(\bigcup\limits_{t\in\mathcal{T}_j}a_j^T\pmb{x}(t)\geq b_j-a_j^Tx^{plan}(t)\right).
    \end{aligned}
\end{equation}
Two equalities of (\ref{17_LHS}) follow from (\ref{Hj+-}) and (\ref{deviation trajectory}) respectively. Now, recall that $d_j$ is the minimum distance of $\mathcal{S}_j$ from $\mathcal{X}_{obs}$ i.e., $d_j=b_j-a_j^Ty_2^*$, where $y_2^*$ is the solution to the optimization problem (\ref{optimization problem}). Noting that $b_j-a_j^Tx^{plan}(t)\geq d_j$, $\forall\, t\in \mathcal{T}_j$, from (\ref{17_LHS}),
\begin{equation*}
\left(\bigcup\limits_{t\in\mathcal{T}_j}\pmb{x}^{sys}(t)\in \mathcal{H}_j^+\!\! \right)\subseteq\left(\bigcup\limits_{t\in\mathcal{T}_j}a_j^T\pmb{x}(t)\geq d_j\!\right).
\end{equation*}
\section*{APPENDIX B}
\section*{Computation of $P(\mathcal{D}_j)$ and $P(\mathcal{D}_j\cap \mathcal{D}_{j+1})$}
Let us define: $\Delta\hat{t}_j^i\coloneqq\hat{t}_j^{i+1}-\hat{t}_j^i$, and $\hat{\pmb{x}}_j^i\coloneqq \pmb{x}(\hat{t}_j^i)$.
From (\ref{xj dynamics}), we can write 
\begin{equation}\label{seg dynamics}
\hat{\pmb{x}}_j^{i+1}=\hat{\pmb{x}}_j^i+\hat{\pmb{n}}_j^i, \qquad \hat{\pmb{n}}_j^i\sim\mathcal{N}(0, \Sigma_{\hat{\pmb{n}}_j^i})  
\end{equation}
where $\hat{\pmb{x}}_j^0=\pmb{x}_{j-1}$, $\Sigma_{\hat{\pmb{n}}_j^i}\coloneqq\Delta\hat{t}_j^iR$, for $i=0, 1, \hdots, r_j-1$, and $j=1, 2, \hdots, N$. Multiplying both sides of (\ref{seg dynamics}) by $a_j^T$ we get
\begin{equation*}
a_j^T\hat{\pmb{x}}_j^{i+1}=\pmb{z}^{i+1}_j=a_j^T\hat{\pmb{x}}_j^i+a_j^T\hat{\pmb{n}}_j^i, \qquad \hat{\pmb{n}}_j^i\sim\mathcal{N}(0, \Sigma_{\hat{\pmb{n}}_j^i}).     
\end{equation*}
Stacking all $\pmb{z}_j^{i}$ for $i=0, 1, \hdots, r_j$, we can write the dynamics for the entire time segment $\mathcal{T}_j$ as
\begin{equation}\label{xj_seg}
    \pmb{z}_j^{seg}=M_j\hat{\pmb{x}}_j^{0}+K_j\hat{\pmb{n}}_j^{seg}, \qquad \hat{\pmb{n}}_j^{seg}\sim\mathcal{N}(0, \Sigma_{\hat{\pmb{n}}_j^{seg}})
\end{equation}
where $\pmb{z}_j^{seg}\coloneqq\begin{bmatrix} \pmb{z}_j^{0}& \pmb{z}_j^{1}& \hdots &\pmb{z}_j^{r_j} \end{bmatrix}^T$, $M_j=a_j^T\cdot\mathds{1}$,\\
$
\hat{\pmb{n}}_j^{seg}=\begin{bmatrix}
\hat{\pmb{n}}_j^{0}&
\hat{\pmb{n}}_j^{1}&
\hdots&
\hat{\pmb{n}}_j^{r_j-1}
\end{bmatrix}^T, \quad \Sigma_{\hat{\pmb{n}}_j^{seg}}=\underset{0 \leq i \leq r_j-1}{\text{diag}}\Sigma_{\hat{\pmb{n}}_j^i},
$
\begin{equation*}
K_j=\begin{bmatrix}
0&0&\hdots&0\\
a_j^T&0&\hdots&0\\
a_j^T&a_j^T&\hdots&0\\
\vdots&\vdots&\ddots&\vdots\\
a_j^T&a_j^T&\hdots&a_j^T\end{bmatrix}    
\end{equation*}
\subsection*{Computation of $P(\mathcal{D}_j)$:}
In order to compute $P(\mathcal{D}_j)$, we need to find the distribution of $\pmb{z}_j^{seg}$. Since $\hat{\pmb{x}}_j^{0}=\pmb{x}_{j-1}$, it is distributed as $\hat{\pmb{x}}_j^{0}\sim\mathcal{N}(0, \Sigma_{\pmb{x}_{j-1}})$. Hence, from (\ref{xj_seg}), the p.d.f. of $\pmb{z}_j^{seg}$ can be written as $\mu_{\pmb{z}_j^{seg}}({z}_j^{seg})=\mathcal{N}(0, \Sigma_{\pmb{z}_j^{seg}})$ where $\Sigma_{\pmb{z}_j^{seg}}=M_j\Sigma_{\pmb{x}_{j-1}}M_j^T+K_j\Sigma_{\hat{\pmb{n}}_j^{seg}}K_j^T$. Now, $P(\mathcal{D}_j)$ can be computed as
\begin{equation}\label{P(Dj)}
    P(\mathcal{D}_j)=\int_{\mathcal{C}_j} \mu_{\pmb{z}_j^{seg}}({z}_j^{seg}) \,d{z}_j^{seg}
\end{equation}
where $\mathcal{C}_j$ is a hypercube of dimension $r_j+1$, having its sides along each direction run from $-\infty$ to $d_j$. 
\subsection*{Computation of $P(\mathcal{D}_j\cap \mathcal{D}_{j+1})$:}
Let us define $\pmb{z}_{j,\,j+1}^{seg}\coloneqq\begin{bmatrix} \pmb{z}_{j}^{seg} & \pmb{z}_{j+1}^{seg} \end{bmatrix}^T$. In order to compute $P(\mathcal{D}_j\cap \mathcal{D}_{j+1})$, we need to find the distribution of $\pmb{z}_{j,\,j+1}^{seg}$. First, let us write $\hat{\pmb{x}}_{j+1}^{0}$ in terms of $\hat{\pmb{x}}_{j}^{0}$.
\begin{equation}\label{xk in xj}
\hat{\pmb{x}}_{j+1}^{0}=\hat{\pmb{x}}_{j}^{0}+
G_j\hat{\pmb{n}}_{j}^{seg}
\end{equation}
where $G_j=\begin{bmatrix}
    I & I &\hdots & I
    \end{bmatrix}_{n\times nr_j}$, $j=1, 2, \hdots, N-1,$ and $I$ is an $n\times n$ identity matrix. We know that
\begin{equation}\label{xk_seg}
    \pmb{z}_{j+1}^{seg}=M_{j+1}\hat{\pmb{x}}_{j+1}^{0}+K_{j+1}\hat{\pmb{n}}_{j+1}^{seg}.
\end{equation}
Substituting $\hat{\pmb{x}}_{j+1}^{0}$ from (\ref{xk in xj}) in (\ref{xk_seg}), we get
\begin{equation}\label{xk_seg in xj_0}
\begin{aligned}
\pmb{z}_{j+1}^{seg}=M_{j+1}\hat{\pmb{x}}_j^{0}+M_{j+1}G_j\hat{\pmb{n}}_j^{seg}+K_{j+1}\hat{\pmb{n}}_{j+1}^{seg}.  
\end{aligned} 
\end{equation}
Let $H_{j,\,j+1}\coloneqq cov(\pmb{z}_{j}^{seg},\;\pmb{z}_{j+1}^{seg})$.
Using (\ref{xj_seg}) and (\ref{xk_seg in xj_0}), we can show that
\begin{equation}\label{Hjk}
H_{j,\,j+1}=M_j\Sigma_{\pmb{x}_{j-1}}M_{j+1}^T+K_j\Sigma_{\hat{\pmb{n}}_j^{seg}}G_j^TM_{j+1}^T.    
\end{equation}
For computing (\ref{Hjk}) we use the fact that
\begin{equation*}
    cov(\hat{\pmb{x}}_j^{0}, \hat{\pmb{n}}_j^{seg})=cov(\hat{\pmb{x}}_j^{0}, \hat{\pmb{n}}_{j+1}^{seg})=cov(\hat{\pmb{n}}_{j}^{seg},\hat{\pmb{n}}_{j+1}^{seg})=0
\end{equation*}
Now, the p.d.f. of $\pmb{z}_{j,\,j+1}^{seg}$ can be written as $ \mu_{\pmb{z}_{j,\,j+1}^{seg}}\left({z}_{j,\,j+1}^{seg}\right)=\mathcal{N}(0, \Sigma_{\pmb{z}_{j,\,j+1}^{seg}})$ where
\begin{equation*}
   \Sigma_{\pmb{z}_{j,\,j+1}^{seg}}=\begin{bmatrix}
   \Sigma_{\pmb{z}_{j}^{seg}} & H_{j,\,j+1}\\
    H_{j,\,j+1}^T & \Sigma_{\pmb{z}_{j+1}^{seg}}
    \end{bmatrix}
\end{equation*}
and $P(\mathcal{D}_j\cap \mathcal{D}_{j+1})$ can be computed as
\begin{equation}\label{P(Dj,Dk)}
P(\mathcal{D}_j\cap \mathcal{D}_{j+1})=\!\int_{\mathcal{C}_j}\!\int_{\mathcal{C}_{j+1}} \!\!\!\!\mu_{\pmb{z}_{j,\,j+1}^{seg}}\left({z}_{j,\,j+1}^{seg}\right) \,d{z}_{j+1}^{seg}d{z}_{j}^{seg}.  
\end{equation}
MATLAB's \texttt{mvncdf} function can be utilized for computing (\ref{P(Dj)}) and (\ref{P(Dj,Dk)}) numerically.
\bibliographystyle{IEEEtran}
\bibliography{bibliography}
\end{document}